\documentclass[10pt,twocolumn,twoside]{IEEEtran}
\usepackage{graphicx}
\usepackage{graphics} 
\usepackage{epsfig} 
\usepackage[abs]{overpic}
\usepackage{algorithmic}
\usepackage{epstopdf}
\usepackage{amsmath}
\usepackage{amsfonts}

\usepackage{tikz}
\usepackage{color}
\usepackage{amsthm} 
\newtheorem{theorem}{\bf{Theorem}}[section]

\newtheorem{lem}[theorem]{Lemma}

\theoremstyle{plain}

\newenvironment{definition}[1][Definition]{\begin{trivlist}
\item[\hskip \labelsep {\bfseries #1}]}{\end{trivlist}}

\newcounter{assump} 
\newcounter{rem} 
\newcounter{propno} 
\newenvironment{property}[1][Property \arabic{propno}]{\refstepcounter{propno} \begin{trivlist} 
\item[\hskip \labelsep {\bfseries #1}]}{\end{trivlist}}

\newcounter{algno} 
\newenvironment{alg}{\refstepcounter{algno}\linespread{1.5}\small\tabular}{\endtabular}

\title{\LARGE \bf
Communication-Free Distributed Coverage for Networked Systems
}

\author{A. Yasin Yaz{\i}c{\i}o\u{g}lu, Magnus Egerstedt, and Jeff S. Shamma
\thanks{This work was supported by ONR project \#N00014-09-1-0751. \newline \indent A. Yasin Yaz{\i}c{\i}o\u{g}lu is with the Laboratory for Information and Decision Systems, Massachusetts Institute of Technology, {\tt yasiny@mit.edu}. \newline \indent
Magnus Egerstedt is with the School of Electrical and Computer Engineering, Georgia Institute of Technology, {\tt magnus@gatech.edu}.\newline \indent
Jeff S. Shamma is with the School of Electrical and Computer Engineering, Georgia Institute of Technology, {\tt shamma@gatech.edu}, and with King Abdullah University of Science and Technology (KAUST), {\tt jeff.shamma@kaust.edu.sa.}}}
\IEEEoverridecommandlockouts
\begin{document}

\maketitle
\begin{abstract}
In this paper, we present a communication-free algorithm for distributed coverage of an arbitrary network by a group of mobile agents with local sensing capabilities. The network is represented as a graph, and the agents are arbitrarily deployed on some nodes of the graph. Any node of the graph is covered if it is within the sensing range of at least one agent. The agents are mobile devices that aim to explore the graph and to optimize their locations in a decentralized fashion by relying only on their sensory inputs. We formulate this problem in a game theoretic setting and propose a communication-free learning algorithm for maximizing the coverage. 

\end{abstract}


\section{Introduction}

 In many networked systems, a typical task is to provide some service such as security or maintenance via some agents with limited capabilities (e.g., \cite{Goddard05, Du03, Berbeglia10}). One way of achieving this task is to solve a locational optimization problem (e.g., \cite{Reese06, Megiddo83,Khuller99, Owen98, Caprara2000}) and let each agent serve some part of the network around its assigned position. In the absence of a centralized mechanism, the agents are faced with a distributed coverage control problem, where their objective is to optimize their locations by following some decentralized controllers.  
 
 

 
 
 Distributed coverage control is widely studied on continuous domains (e.g., \cite{Howard02}-\cite{Schwager09}).
One possible approach is to employ potential fields to drive each agent away from the nearby agents and obstacles (e.g., \cite{Howard02,Poduri04}). Alternatively, a prevailing approach introduced in \cite{Cortes04} is to model the underlying locational optimization problem as a continuous $p$-median problem and to employ Lloyd's algorithm \cite{Lloyd82}. As such, the agents are driven onto a local optimum, i.e. a centroidal Voronoi configuration, where each point in the space is assigned to the nearest agent, and each agent is located at the center of mass of its own region.
Later on, this method was extended for agents with distance-limited sensing and communications (e.g.,  \cite{Cortes05}) and limited power (e.g., \cite{Kwok10}), as well as for heterogeneous agents covering non-convex regions (e.g., \cite{Pimenta08}). Also, the requirement of sensing density functions was relaxed by incorporating methods from adaptive control and learning (e.g., \cite{Schwager09}). 

In some studies, distributed coverage control was studied on discrete spaces represented as graphs (e.g., \cite{Durham09,Yun12,Zhu13,Yasin13NecSys}). One possible approach is to achieve a centroidal Voronoi partition of the graph via pairwise gossip algorithms (e.g., \cite{Durham09}) or via asynchronous greedy updates (e.g., \cite{Yun12}). Alternatively, distributed coverage control on discrete spaces can be studied in a game theoretic framework (e.g., \cite{Zhu13,Yasin13NecSys}). Game theoretic methods have been used to solve many cooperative control problems such as vehicle-target assignment (e.g., \cite{Arslan07}), dynamic vehicle routing (e.g. \cite{Arsie09}), cooperative communication (e.g., \cite{Huang08}), and coverage optimization (e.g., \cite{Zhu13,Yasin13NecSys}). In \cite{Zhu13}, sensors with variable footprints achieve power-aware optimal coverage on a discretized space. In \cite{Yasin13NecSys}, a group of heterogeneous mobile agents are driven on a graph to maximize the number of covered nodes. 

In this paper, we study a distributed coverage control problem on graphs in a game theoretic setting. In this problem, mobile agents are arbitrarily deployed on an unknown graph. Each agent is assumed to sense the local graph structure and the presence of other agents (if any) within its sensing range. Any node of the graph is covered if it is within the sensing range of at least one agent. The objective of the agents is to maximize the number of covered nodes by optimizing their locations on the graph.  We present a game theoretic formulation for this coverage control problem. We particularly focus on a communication-free setting, where each agent should be driven via only its sensory inputs. In that case, the agents do not observe their exact utilities in the corresponding game. Accordingly, we propose a learning algorithm for driving the agent positions based on some estimated utilities. Using the proposed method, the agents maintain optimal coverage with an arbitrarily high probability as time goes to infinity.

The organization of this paper is as follows: Section \ref{cover} presents the distributed graph coverage problem. Section \ref{game} sets up the game-theoretic formulation of the problem. Section \ref{covmax} presents a solution that requires some explicit communications among the agents. The proposed communication-free solution is presented in Section \ref{sol}. Some simulation results for the proposed method are presented in Section \ref{sims}. Finally, Section \ref{conclusion} concludes the paper. 

\section{Distributed Graph Coverage}
\label{cover}
In this section, we present the distributed graph coverage (DGC) problem, where the goal is to maximize the number of covered nodes by driving the agents with limited sensing and mobility capabilities to optimal locations on a graph. First, some graph theory preliminaries are presented. 

\subsection{Graph Theory Concepts}
An undirected graph $\mathcal{G}=(V,E)$ consists of a \emph{node set} $V$ and an \emph{edge set} $E \subseteq V \times V$. For an undirected graph, the edge set consists of unordered node pairs $(v,v')$ denoting that the nodes $v$ and $v'$ are adjacent. 

A \emph{path} is a sequence of nodes such that each node is adjacent to the preceding node in the sequence. For any two nodes $v$ and $v'$, the \emph{distance} between the nodes $d(v,v')$ is the number of edges in a shortest path between $v$ and $v'$. A graph is \emph{connected} if the distance between any pair of nodes is finite. 

The set of nodes containing a node $v$ and all the nodes adjacent to $v$ is called the \emph{(closed) neighborhood} of $v$, and it is denoted as $\mathcal{N}_v$. For any $\delta \geq 0$, the $\delta$-neighborhood of $v$, $\mathcal{N}^{\delta}_{v}$, is the set of nodes that are at most $\delta$ away from $v$, i.e.

\begin{equation}
\label{dneigh}
\mathcal{N}^{\delta}_v = \{v' \in V \mid d(v,v') \leq \delta\}.
\end{equation}

For any $\mathcal{G}=(V,E)$,  an \emph{induced subgraph}, $\mathcal{G}[V_s]$, consists of the vertices, $V_s \subseteq V$, and the edges whose endpoints are both in $V_s$. 


\subsection{Problem Formulation}



Consider a connected undirected graph, $\mathcal{G} = (V,E)$, and let $I=\{1, 2, \hdots, m\}$ denote a set of $m$ mobile agents arbitrarily deployed on some nodes of the graph.  Let each agent have a sensing range, $\delta$. We assume that each agent, $i$, can sense the subgraph induced by the nodes in $\mathcal{N}^{\delta}_{v_i}$ and the presence of other agents (if any) within its $\delta$-neighborhood. As such, each $i  \in I$ located at $v_i \in V$ covers all the nodes in $\mathcal{N}^{\delta}_{v_i}$. Any node of the graph is covered if it is included in the $\delta$-neighborhood of at least one agent, and the set of covered nodes, $V_c\subseteq V$, is given as
\begin{equation}
\label{Vc}
V_c (v_1, \hdots, v_m)=\bigcup_{i=1}^m\mathcal{N}^{\delta}_{v_i}.
\end{equation} 

The objective in the distributed graph coverage (DGC) problem is to have the agents update their positions over time in a distributed manner to maximize the number of covered nodes, i.e.
\begin{equation}
\label{objective}
|V_c (v_1(t), \hdots, v_m(t))|,
\end{equation} 
where each $v_i(t) \in V$ is the position of agent $i$ at time $t$.

In order to achieve optimal coverage in a distributed fashion, the agents need some local rules to follow. In general, a rule is considered to be local if its execution by an agent requires only some information available within a small distance from the agent. In this paper, we consider a discrete time dynamics and we assume that each agent can either maintain its position or move to an adjacent node in the next time step, i.e.
\begin{equation}
\label{cmove}
d(v_i(t+1),v_i(t)) \leq 1, \;  \forall i \in \{1,2, \hdots, m\}.
\end{equation} 





\subsection{Solution Approach}
In the DGC problem,  a group of mobile agents explore an unknown graph and aim to cover as many nodes as possible. As such, the underlying locational optimization problem is similar to the maximum coverage problem (e.g., \cite{Megiddo83,Khuller99}). Such NP-hard problems are typically tackled by finding sufficiently good approximate solutions through fast algorithms (e.g. \cite{Jia02,Kuhn05,Abrams04}). Similarly, in many distributed coverage control studies, a locational objective function is optimized by the agents aiming for the best local improvements (e.g., \cite{Cortes04}-\cite{Yun12}). Such a distributed greedy approach  can be employed to solve the DGC problem. Accordingly, the agents may move locally on the graph to maximally improve their local coverage. In that case, the resulting performance would significantly depend on the graph structure and the initial configuration. This method may rapidly lead to a reasonable approximate solution if the agents start with a sufficiently good initial coverage or if the interaction graph satisfies some structural properties. However, it may also lead to arbitrarily poor configurations for arbitrary graphs and initial conditions. For instance, consider the scenario in Fig. \ref{counterex}, where 2 agents with sensing ranges $\delta=1$ can achieve a globally optimal configuration in 2 time steps.  In this example, the initial configuration would be stationary under a greedy approach since none of the agents can improve the coverage by moving to a neighboring node. Note that the performance in Fig. \ref{counterex}a would be arbitrarily poor for any arbitrarily large graph obtained by adding more leaf nodes attached to the unoccupied hub.

\begin{figure}[htb]
\centering
\includegraphics[trim =0mm 0mm 0mm 0mm, clip,scale=0.68]{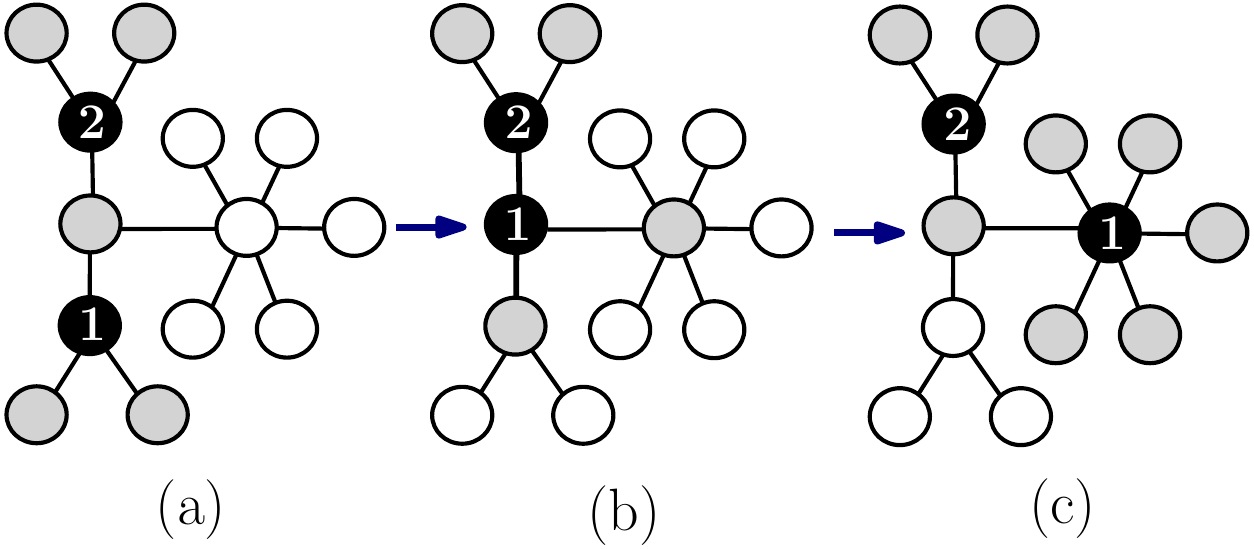}
\caption{A possible trajectory to a globally optimal configuration for two agents on a small graph. The agents have cover ranges $\delta=1$, and they are initially located as in $(a)$. The number of covered nodes (shown as non-white) is reduced in the intermediate step illustrated in (b) to reach the global optima shown in (c).}
\label{counterex}
\end{figure}

In order to ensure efficient coverage for arbitrary graphs and initial configurations, a solution method should occasionally allow for graph exploration at the expense of a better coverage. In this work, we present such a solution by approaching the problem from a game theoretic perspective. In particular, we map the DGC problem to a game, and we design a learning algorithm for the agents to follow in updating their actions.

\section{Game Theoretic Formulation}
\label{game}
In this section, a game-theoretic formulation of the DGC problem is presented. First, some game theory preliminaries are provided.

\subsection{Game Theory Concepts}
A finite \emph{strategic game} $\Gamma= (I, A, U)$ consists of three components: (1) a set of $m$ \emph{players (agents)} $I=\{1, 2, \hdots, m\}$, (2) an $m$-dimensional \emph{action space} $A= A_1 \times A_2 \times \hdots\ \times A_m$, where each $A_i$ is the \emph{action set} of player $i$, and (3) a set of \emph{utility functions} $U= \{U_1, U_2, \hdots , U_m\}$, where each $U_i : A \mapsto \Re$ is a mapping from the action space to real numbers.   

For any \emph{action profile} $a \in A$, let $a_{-i}$ denote the actions of players other than $i$. Using this notation, an action profile $a$ can also be represented as $a=(a_i,a_{-i})$.




A class of games that is widely utilized in cooperative control problems is potential games. A game is called a potential game if there exists a \emph{potential function}, $\phi : A \mapsto \Re$, such that the change of a player's utility resulting form its unilateral deviation from an action profile equals the resulting change in $\phi$. More precisely, for each player $i$, for every $a_i, a'_i \in A_i$, and for all $a_{-i} \in A_{-i}$,  
\begin{equation}
\label{PG}
 U_i(a'_i, a_{-i})- U_i(a_i, a_{-i}) = \phi(a'_i, a_{-i})- \phi(a_i, a_{-i}).
\end{equation}
When a cooperative control problem is mapped to a potential game, usually the game is designed such that its potential function captures the global objective of the control problem. Once a such potential game is designed, some game theoretic learning algorithms such as log-linear learning \cite{Blume93} can be utilized to drive the agent actions to the set of potential maximizers. 






\subsection{DGC Game}
In order to formulate the DGC problem in a game theoretic setting, we design a corresponding game, $\Gamma_{\text{DGC}}$, by defining the action space and the utility functions. More specifically, we design a potential game such that its potential function, $\phi(a)$, captures the global objective of the DGC problem, i.e. 

\begin{equation}
\label{potf}
\phi(a)= |V_c(a)|.
\end{equation}

In the DGC problem, the coverage provided by each agent is determined by the position of the agent. Hence, the action of an agent can be defined as its position on the graph. Accordingly, each action set is equal to the node set of $\mathcal{G}=(V,E)$, i.e.
\begin{equation}
\label{actionset}
A_i= V, \quad \forall i \in I. 
\end{equation}

Then, the utilities should be designed such that $\phi(a)$ in (\ref{potf}) is indeed the potential function for the resulting game. To this end, we design the agent utilities as

\begin{eqnarray}
\label{util}
U_i (a) &=&  |\mathcal{N}_{a_i}^{\delta} \setminus \bigcup_{j \neq i}\mathcal{N}_{a_j}^{\delta}|, \nonumber \\
&=& \sum_{v \in \mathcal{N}_{a_i}^{\delta}} u_i(v, a_{-i}),
\end{eqnarray}
where, for every $v \in \mathcal{N}_{a_i}^{\delta}$, $u_i(v, a_{-i})$ is the partial utility agent $i$ gathers by covering node $v$, and it is defined as
\begin{equation}
\label{smallu}
u_i(v, a_{-i})=  \left\{\begin{array}{ll} 1&\mbox{ if $d(v,a_j) > \delta$ \; $\forall j\neq i$}, \\ 0&\mbox{ otherwise. }\end{array}\right.
\end{equation} 

In the resulting game, each agent gathers a utility equal to the number of nodes that are covered only by itself. Note that this utility is equal to the marginal contribution of the corresponding agent to the number of covered nodes.  
%


\begin{lem}
\label{potgam}
The utilities in (\ref{util}) lead to a potential game $\Gamma_{\emph{DGC}}= (P, A, U)$ with the potential function given in (\ref{potf}).
\end{lem}

\begin{proof}
Let $a_i =v_i$ and $a_i'=v_i'$ be two possible actions for agent $i$, and let $a_{-i}$ denote the actions of other agents. Due to (\ref{Vc}) and (\ref{potf}),
 \begin{equation}
\label{exppot}
\phi(a)= |\bigcup_{i \in I}\mathcal{N}^{\delta}_{a_i}|\end{equation}
Using (\ref{util}), for any agent $i$, (\ref{exppot}) can be expanded as
 \begin{equation}
\label{exppot2}
\phi(a) = |\mathcal{N}^{\delta}_{a_i} \setminus \bigcup_{j \neq i}\mathcal{N}^{\delta}_{a_j}|+ |\bigcup_{j \neq i}\mathcal{N}^{\delta}_{a_j}| 
= U_i (a_i,a_{-i})+ |\bigcup_{j \neq i} \mathcal{N}^{\delta_j}_{a_j}|.
\end{equation}
Using (\ref{exppot2}) for any pair of actions $a_i$ and $a_i'$,
 \begin{equation}
\label{potpr}
\phi(a'_i, a_{-i})- \phi(a_i, a_{-i})=U_i (a_i',a_{-i})- U_i (a_i,a_{-i}).
\end{equation}
\end{proof}

\subsection{Learning}
In game theoretic learning, starting from an arbitrary initial configuration, the agents repetitively play a game. At each step $t \in \{0, 1, 2, \hdots \}$, each agent $i \in  I$ plays an action $a_i(t)$ and receives some utility $U_i(a(t))$. In this setting, the agents update their actions in accordance with some learning algorithms. For the DGC problem, the learning process is desired to drive the agent positions to the set of configurations that maximize the number of covered nodes. 

For potential games, a learning algorithm known as \emph{log-linear learning} (LLL) can be used to drive the agents to action profiles that maximize the potential function $\phi(a)$ \cite{Blume93}. Essentially, LLL is a noisy best-response algorithm, and it induces a Markov chain over the action space with a unique limiting distribution, $\mu^*_\epsilon$,  where $\epsilon$ denotes the noise parameter. As the noise parameter, $\epsilon$, goes down to zero, the limiting distribution, $\mu^*_\epsilon$, has an arbitrarily large part of its mass over the set of potential maximizers \cite{Blume93}. However, LLL assumes that at any round each player $i$ has access to all the actions in its action set $A_i$. In general, LLL may not provide potential maximization when the system evolves over constrained action sets, i.e. when each agent $i$ is allowed to choose its next action from only a subset of actions. Note that this is indeed the case for the DGC problem, and each agent has to pick its the next action from the closed neighborhood of its current action $a_i$, 
\begin{equation}
\label{constrA}
A^c_i(a_i)= \mathcal{N}_{a_i} \; \forall i \in I.
\end{equation}
The issue of constrained action sets was addressed in \cite{Marden12}, and a variant learning algorithm called \emph{binary log-linear learning} (BLLL) was presented for such cases. 

In learning algorithms, typically each agent is assumed to measure its current utility. 
For instance, in order to execute LLL or BLLL, the agents need to measure their utilities resulting from their current actions as well as the hypothetical utilities they may gather by unilaterally switching to some other actions. Alternatively, a payoff-based implementation may be utilized to avoid the necessity to compute the hypothetical utilities \cite{Marden12}. Note that, for $\Gamma_{\text{DGC}}$, even the computation of the current utility requires some \textit{explicit communications} since the agents with overlapping coverage are not necessarily within the sensing range of each other. In general, such agents can be up to $2\delta$ apart on the graph.  

\subsection{Stochastic Stability Concepts}
For potential games, noisy best-response algorithms such as LLL or BLLL induce a regular perturbed Markov chain over the action space such that the stochastically stable states are the potential maximizers. The concept of stochastic stability will be extensively used in the remainder of this paper. Hence, we provide some preliminaries prior to our derivations.

\begin{definition} \emph{(Regular Perturbed Markov Chain)}: Let $P$ be the transition matrix of a discrete-time Markov chain over a finite state space $\mathcal{X}$. A perturbed Markov chain with the noise parameter $\epsilon$ is called a regular perturbed Markov chain if
\begin{enumerate}
\item $P_{\epsilon}$ is aperiodic and irreducible for $\epsilon >0$,
\item $\lim_{\epsilon \rightarrow 0} P_{\epsilon}=P$,
\item For any $x,x^+ \in \mathcal{X}$ if $P_{\epsilon}(x,x^+)>0$, then there exists $R(x,x^+) \geq 0$ such that
\begin{equation}
\label{resdef}
0 < \lim_ {\epsilon \rightarrow 0^+} \frac{P_{\epsilon}(x, x^+)}{\epsilon^{R(x,x^+)}} < \infty,
\end{equation}
where $R(x,x^+)$ is called the resistance of the transition from $x$ to $x^+$.
\end{enumerate}
\end{definition}

Any regular perturbed Markov chain, $P_{\epsilon}$, has a unique limiting distribution, $\mu^*_{\epsilon}$, since it is aperiodic and irreducible.
\begin{definition} \emph{(Stochastically Stable State)}: Let $P_{\epsilon}$ denote a regular perturbed Markov chain over a state space, $\mathcal{X}$. Any state, $x\in \mathcal{X}$, is stochastically stable if 
\begin{equation}
\label{sstab}
\lim_ {\epsilon \rightarrow 0^+} \mu^*_{\epsilon}(x) >0.
\end{equation}
\end{definition}

The stochastically stable states of a regular perturbed Markov chain, $P_\epsilon$, can be  characterized through a resistance tree analysis. For any $x \in \mathcal{X}$, a spanning tree rooted at $x$, $\mathcal{T}_x$, is a directed graph, where the nodes correspond to states, directed edges correspond to some feasible state transitions, and there is a unique directed path on $\mathcal{T}_x$ from any state $x' \neq x$ to $x$. The resistance of such a tree, $R(\mathcal{T}_x)$, is defined as the sum of the resistances of its edges, where the resistance of each edge is given as in (\ref{resdef}). $\mathcal{T}^*_x$ is called a minimum resistance tree if  $R(\mathcal{T}^*_x) \leq R(\mathcal{T}_x)$ for any $\mathcal{T}_x$, i.e. any spanning tree rooted at $x$ has at least as much resistance as $\mathcal{T}^*_x$. The stochastic potential of a state, $x$, is defined as the total resistance of its minimum resistance tree, $R(\mathcal{T}^*_x)$. The following result characterizes the stochastically stable states through their stochastic potentials.

\begin{lem} \cite{Young93}
\label{stochrem}
Let $P_\epsilon$ be a regular perturbed Markov chain. Any $x \in \mathcal{X}$ is stochastically stable if and only if $x$ is a recurrent state of the unperturbed chain, $P_0$, with the minimum stochastic potential.
\end{lem}

\section{Coverage Maximization}
\label{covmax}
 In this section, we will briefly show that if all the agents follow BLLL in a repetitive play of $\Gamma_{\text{DGC}}$, then the stochastically stable states are the coverage maximizers. A more detailed presentation of this approach can be found in \cite{Yasin13NecSys}. As stated earlier, this solution requires some local communications among the agents. In the next section, we will present a communication-free learning algorithm that can achieve the same limiting behavior as this method. 
\begin{center}
\resizebox{8.6cm}{!}{
\begin{alg}{|l |}

\hline
\label{BLLLAlg}
\textbf{Algorithm \Roman{algno}:} Binary Log-linear Learning (\cite{Marden12})\\
\hline
\mbox{\small $\;1:\;$}\textbf{initialization:} $\epsilon \in \Re^+$ small, $a\in A$ arbitrary \\
\mbox{\small $\;2:\;$}\textbf{repeat} \\

\mbox{\small $\;3:\;$}\hspace{0.45cm} Pick a random $i \in I$, and a random $a_i' \in A^c_i(a_i)$. \\
\mbox{\small $\;4:\;$}\hspace{0.45cm} Compute $\alpha=\epsilon^{-U_i(a(t))}$,  $\beta=\epsilon^{-U_i(a_i',a_{-i}(t))}$. \\
\mbox{\small $\;5:\;$}\hspace{0.45cm} With probability $\frac{\beta}{\alpha+\beta} $, set $a_i = a_i'$.  \\
\mbox{\small $6:\;$}\textbf{end repeat}\\
\hline
\end{alg}}
\end{center}

In BLLL, a single agent is randomly chosen at each time step. The selection of a single agent at each time step can be achieved (with a very high probability) without a centralized coordination by using methods such as the asynchronous time model proposed in \cite{Boyd06}. The selected agent, assuming that all the other agents are stationary, updates its action depending on its current utility and the hypothetical utility it would receive by playing a random action in its constrained action set. This is illustrated in Fig. \ref{BLLLf}.


\begin{figure}[h!]
\begin{center}
\includegraphics[trim =0mm 0mm 0mm 0mm, clip,scale=0.8]{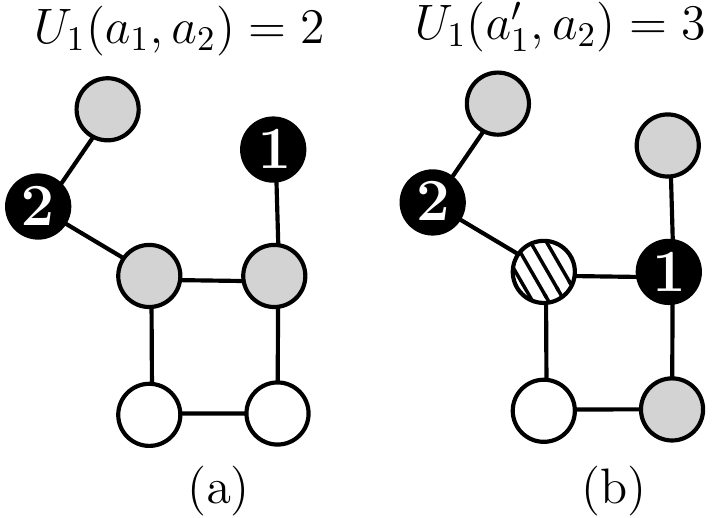}
\caption{An illustration of the BLLL algorithm, where two agents with $\delta=1$ are located as in (a) and Agent 1 is updating its action. Agent 1 randomly picks a candidate action, $a_1' \in A_i^c(a_i)$, as in (b). Its next action is picked from $\{a_1,a_1'\}$ with probabilities depending on the corresponding utilities. For the configuration in (b), the tiled node is not providing any utility to either of the agents since it is covered by both of them. }
\label{BLLLf}
\end{center}
\end{figure}

In \cite{Marden12}, it was shown that BLLL can be used to achieve potential maximization if the constrained action sets satisfy Properties \ref{proper1} and \ref{proper2} provided below.

\begin{property}(Reachability) For any agent $i \in I$ and any action pair $a_i^0, a_i^k \in A_i$, there exists a sequence of actions $\{a_i^0,a_i^1, \hdots, a_i^k\}$ such that $a_i^r \in A^c_i(a_i^{r-1})$ for all $r \in \{1,2, \hdots, k\}$. 
\label{proper1} 
\end{property}
\vskip3ex
\begin{property}\label{proper2} (Reversability) For any agent $i \in I$ and any action pair $a_i, a_i' \in A_i$,
$$
a_i' \in A^c_i(a_i) \Leftrightarrow a_i \in A^c_i(a_i').
$$
\end{property}

\begin{theorem}
\label{BLLLworks} \cite{Marden12}
Consider any finite potential game and constrained
action sets satisfying Properties \ref{proper1} and \ref{proper2}. If all players adhere to BLLL,
then the stochastically stable states are the set of potential maximizers.
\end{theorem}


In light of Theorem \ref{BLLLworks}, the agents can maximize the coverage by following the BLLL algorithm in a repetitive play of $\Gamma_{\text{DGC}}$, if the constrained action sets given in (\ref{constrA}) satisfy Properties \ref{proper1} and \ref{proper2}. Lemma \ref{Propsatis} shows that the constrained action sets indeed satisfy these properties if the graph to be covered is connected. 

\begin{lem}
\label{Propsatis}
The constrained action sets in (\ref{constrA}) satisfy Properties \ref{proper1} and \ref{proper2} if the graph $\mathcal{G}=(V,E)$ is connected. 
\end{lem}
\begin{proof}
If the graph is connected, then there exists a finite-length path $\{v^0, \hdots, v^k\}$ between any pair of nodes $v^0,v^k \in V$, and Property \ref{proper1} is satisfied. Furthermore, for undirected graphs, $d(v,v')=d(v',v)$. Hence, Property \ref{proper2} is also satisfied.
\end{proof}

\begin{theorem}
\label{Solved}
Let  $\mathcal{G}=(V,E)$ be connected graph, and let all agents follow BLLL in a repetitive play of $\Gamma_{\emph{DGC}}$ with the constrained action sets in (\ref{constrA}). Then the stochastically stable states are the maximizers of $|V_c(a)|$.

\end{theorem}
\begin{proof}
If  $\mathcal{G}=(V,E)$ is connected, then the constrained action sets in (\ref{constrA}) satisfy Properties \ref{proper1} and \ref{proper2} due to Lemma \ref{Propsatis}. Hence, in light of Theorem \ref{BLLLworks}, if all agents follow BLLL in a repetitive play of $\Gamma_{\text{DGC}}$, the stochastically stable states are the potential maximizers. Due to (\ref{potf}), those are the configurations maximizing the number of covered nodes, $|V_c(a)|$. 
\end{proof}



\section{Communication-free coverage Maximization}
\label{sol}

In the DGC problem, the sensory inputs do not reveal which of the nodes within the sensing range of an agent is covered also by some other agents. However, each agent can sense if any other agent is also covering its current position as illustrated in Fig. \ref{comfree_cov}. Hence, each agent $i$ observes the partial utility, $u_i(a_i, a_{-i})$ in (\ref{smallu}), via its sensory input.



\begin{figure}[h!]
\begin{center}
\includegraphics[trim =0mm 0mm 0mm 0mm, clip,scale=0.52]{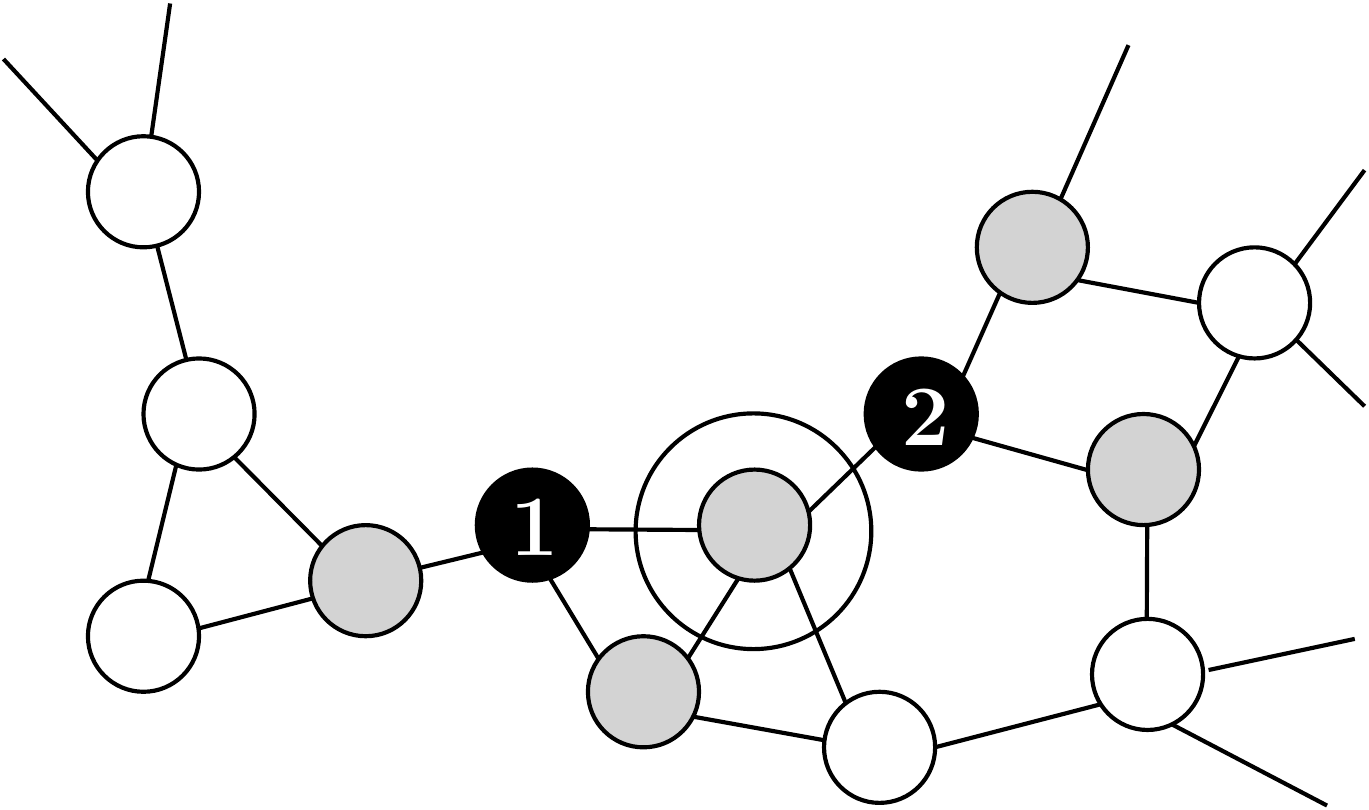}
\caption{ Distributed graph coverage by agents with sensing ranges $\delta=1$. Agents 1 and 2 do not observe that the encircled node is covered by both of them. However, each of them knows that its current position is covered only by itself since no other agent is within its sensing range.  
}
\label{comfree_cov}
\end{center}
\end{figure}

Since the exact utilities in $\Gamma_{\text{DGC}}$ are not measurable without explicit communications, the agents need to update their actions based on some estimated utilities. Assuming that the nearby agents will remain stationary for a sufficient amount of time, each agent $i$ can construct an estimated utility by visiting each $v\in \mathcal{N}_{a_i}^{\delta}$ and combining the sampled  $u_i(v, a_{-i})$. Note that the resulting estimation will not necessarily be equal to the actual utility since multiple agents may be moving simultaneously as illustrated in Fig. \ref{util_est}. However, if the probability of having simultaneously moving agents is sufficiently small, then false estimations will be sufficiently rare for the agents to achieve the desired limiting behavior by following a noisy best-response based on the estimated utilities. The proposed communication-free algorithm is based on this approach.




\begin{figure}[htb]
\centering
\includegraphics[trim =0mm 0mm 0mm 0mm, clip,scale=0.5]{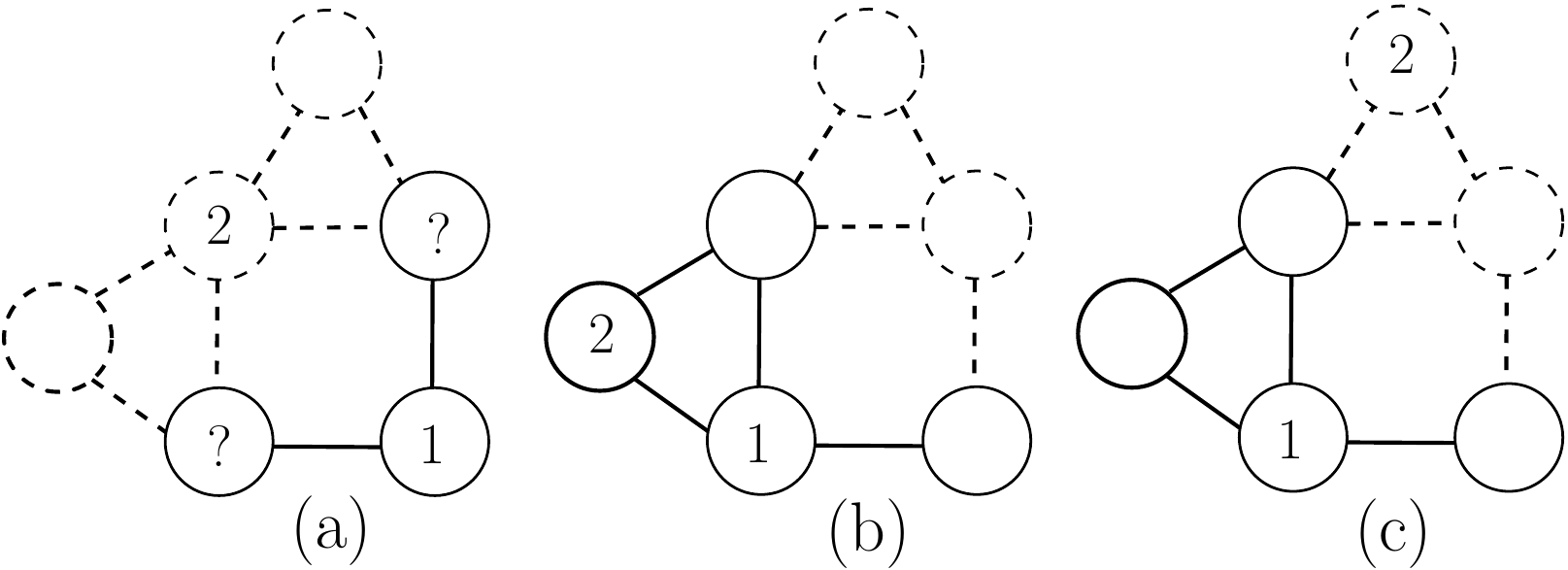}
\caption{Two agents with sensing ranges $\delta=1$ are located on a graph as in (a). Part of the graph that is not sensed by agent 1 is dashed in the figures.  Agent 1 can estimate its utility in (a) by sampling the partial utilities from the nodes in its sensing range. If agent 2 is stationary in the meantime, then the resulting estimation will be true. However, if agent 2 is also moving, then the sampled partial utilities may be true as in (b) or false as in (c). }
\label{util_est}
\end{figure}
 
In the remainder of this section, we present the proposed communication-free coverage maximization algorithm (CFCM) and an analysis of the corresponding dynamics. 
\subsection{CFCM Algorithm} 
The proposed algorithm has two parameters to be set. The first parameter, $\epsilon \in \Re^+$, is the noise in the agent decisions when choosing between the candidate actions based on the corresponding estimated utilities. The second parameter, $r \in \Re^+$, sets the likelihood of each agent to update its action. As it will be shown later in this section, the desired global behavior emerges when $r$ is sufficiently large and $\epsilon$ is small. 





In CFCM, each agent $i$ is either stationary or experimenting. Each stationary agent repeats its current action in the next time step with a high probability, $1-\epsilon^r$, or starts an experiment with probability $\epsilon^r$. An experiment involves comparing its current action, $a_i^1$, to an alternative randomly picked from its constrained action set, $a_i^2 \in  A^c_i(a_i^1) $, where $ A^c_i(a_i^1)$ is the local neighborhood of $a_i^1$ as given in (\ref{constrA}). In this aspect, the agent behavior is similar to the payoff-based BLLL in \cite{Marden12}. However, since the agents receive only some partial utilities, $u_i(a_i, a_{-i})$, an experiment consists of visiting all the nodes in $\mathcal{N}^{\delta}_{a_i^1} \cup \mathcal{N}^{\delta}_{a_i^2}$ to see which of those nodes are also covered by some other agents. We refer to the corresponding path to be traversed as an experiment path between $a_i^1$ and $a_i^2$. 





\begin{definition} \emph{(Experiment Path)}: Let $\delta$ be the sensing range of the agents. For any $a_i^1$ and $a_i^2 \in A^c_i(a_i^1)$, a finite path, $\{a_i^1, \hdots, a_i^2\}$, is an experiment path if it traverses $\mathcal{N}^{\delta}_{a_i^1} \cup \mathcal{N}^{\delta}_{a_i^2}$.
\end{definition}
For any $a_i^1$ and $a_i^2 \in A^c_i(a_i^1)$, an experiment path can be obtained locally by utilizing methods such as depth-first search or breadth-first search (e.g., \cite{Tarjan72}). In the CFCM algorithm, an experiment path between $a_i^1$ and $a_i^2$ is denoted as $\mathcal{E}(a_i^1, a_i^2)$. During an experiment, the agent traverses its experiment path to construct the estimated utilities, $\hat{U}_i^1$ and $\hat{U}_i^2$, from the sampled partial utilities. For simplicity, a partial utility from a node is sampled only at the last visit to that node during the experiment. As such, if it is the agent's last visit of the current position, $a_i$, and the agent does not sense any other agent within $\delta$, then the utility estimations corresponding to the candidate actions within $\delta$ from $a_i$ are incremented by 1. Once the experiment path is traversed, the agent randomly chooses between the two candidate actions based on the estimated utilities, $\hat{U}_i^1$ and $\hat{U}_i^2$. At the next time step, the agent becomes stationary at its chosen action until it starts a new experiment.

For the CFCM algorithm, the state of any agent $i$ can be defined as 
\begin{equation}
\label{statei}
x_i=[\begin{array}{llll}  S_i & k_i & \hat{U}_i^1 &\hat{U}_i^2\end{array}],
\end{equation} 
where $S_i$ is a sequence of actions, which is either a singleton (stationary) or an experiment path (experimenting), $k_i \in \{1, \hdots, |S_i|\}$ is an index variable denoting which action in $S_i$ is currently taken by the agent, and $\hat{U}_i^1, \hat{U}_i^2$ are the estimations for $U_i(a_i^1, a_{-i})$ and $U_i(a_i^2, a_{-i})$, respectively. In this representation, the current action, $a_i$, and the candidate actions, $a_i^1$ and $a_i^2$, are given as
\begin{equation}
\label{candact}
a_i=S_i(k_i),\quad
a_i^1=S_i(1), \quad
a_i^2=S_i(|S_i|),
\end{equation}
where $S_i(k_i)$ denotes the $k_i^{th}$ element in $S_i$, and $|S_i|$ denotes the length of $S_i$. 

\vskip2ex
\begin{center}
\resizebox{8.9cm}{!}{
\begin{alg}{|l |}
\hline
\label{CFCMAlg}
\textbf{Algorithm \Roman{algno}:} Communication-free Coverage Maximization (CFCM)\\
\hline
\mbox{\small $\;1:\;$}\textbf{initialization:} $\epsilon \in \Re^+$ (small), $r \in \Re^+$, $a_i\in A_i$ arbitrary, \\
\hspace{0.35cm}   $S_i=\{a_i\}$, $k_i=1$, $\hat{U}_i^1= \hat{U}_i^2= 0$. \\
\mbox{\small $\;2:\;$}\textbf{repeat} \\
\mbox{\small $\;3:\;$}\hspace{0.45cm} $a_i=S_i(k_i)$, \; $a_i^1 = S_i(1)$,   \;$a_i^2=S_i(|S_i|)$.\\
\mbox{\small $\;4:\;$}\hspace{0.45cm} $\textbf{if}\hspace{0.1cm}  $($|S_i|=1$)\\
\mbox{\small $\;5:\;$}\hspace{0.9cm} Generate a random (uniform) $\gamma \in [0,1]$.\\
\mbox{\small $\;6:\;$}\hspace{0.9cm} $\textbf{if}\hspace{0.1cm}  $($\gamma \leq \epsilon^r$)\\
\mbox{\small $\;7:\;$}\hspace{1.35cm} $a_i^2$ is randomly (uniform) chosen over $ A_i^c(a_i^1)$.\\
\mbox{\small $\;8:\;$}\hspace{1.35cm} $S_i=\mathcal{E}(a_i^1,a^2_i)$.\\
\mbox{\small $\;9:\;$}\hspace{0.9cm} $\textbf{end if}$\\
\mbox{\small $10:\;$}\hspace{0.45cm} $\textbf{else}$\\
\mbox{\small $11:\;$}\hspace{0.9cm} $\textbf{if}\hspace{0.1cm}  $$(k_i \geq k, \; \forall k \in \{ k \mid S_i(k) = a_i \})$\\
\mbox{\small $12:\;$}\hspace{1.35cm} $\hat{U}^1_i= \hat{U}^1_i+u_i(a_i,a_{-i})$, if $a_i \in \mathcal{N}_{a^1_i}^\delta$.\\
\mbox{\small $13:\;$}\hspace{1.35cm} $\hat{U}^2_i= \hat{U}^2_i+u_i(a_i,a_{-i})$, if $a_i\in \mathcal{N}_{a^2_i}^\delta$.\\
\mbox{\small $14:\;$}\hspace{0.9cm} \textbf{end if}\\
\mbox{\small $15:\;$}\hspace{0.9cm} $\textbf{if}\hspace{0.1cm}  ( k_i = |S_i|)$\\
\mbox{\small $16:\;$}\hspace{1.35cm} $\alpha=\epsilon^{-\hat{U}_i^1}$, $\beta=\epsilon^{-\hat{U}_i^2}$. \\
\mbox{\small $17:\;$}\hspace{1.35cm} $S_i=\left\{\begin{array}{ll} \{a_i^1\}&\mbox{ w.p. $\frac{\alpha}{\alpha+\beta} $}, \\ \{a_i^2\}&\mbox{ otherwise. }\end{array}\right.$\\
\mbox{\small $18:\;$}\hspace{1.35cm} $k_i=1$, $\hat{U}_i^1=\hat{U}_i^2=0$.\\

\mbox{\small $19:\;$}\hspace{0.9cm} \textbf{else}\\
\mbox{\small $20:\;$}\hspace{1.35cm} $k_i = k_i+1$.\\
\mbox{\small $21:\;$}\hspace{0.9cm} \textbf{end if}\\
\mbox{\small $22:\;$}\hspace{0.45cm} \textbf{end if}\\
\mbox{\small $23:\;$}\textbf{end repeat}\\
\hline
\end{alg}
}
\end{center}
\vskip2ex

The CFCM algorithm is memoryless since the state of every agent in the next time step is independent of its past trajectory. As such, if all agents follow the CFCM algorithm, then a Markov chain is induced over the state space, $\mathcal{X}$, where each $x \in \mathcal{X}$ is the global state obtained by concatenating the states of all agents, i.e.

\begin{equation}
\label{cfstate}
x = [ x_1, x_2, \hdots, x_m].
\end{equation}

In the remainder of this section, the limiting behavior of the resulting Markov chain will be inspected through a stochastic stability analysis.


\subsection{Limiting Behavior} 
%
%
%

For any $x \in \mathcal{X}$, the agents can be grouped into two distinct sets consisting of the stationary agents, $I_s(x)$, and the experimenting agents, $I_e(x)$, as
\begin{equation}
I_s(x)= \{i \in I \mid |S_i|=1\},
\end{equation}
\begin{equation}\label{Ie}
I_e(x)= I\setminus I_s(x).
\end{equation}
Using these sets, for any feasible transition, $x \rightarrow x^+$, the agents can be grouped into 4 disjoint sets based on the transition of their individual states:
\begin{equation}
I_{ss}(x,x^+)= I_s(x) \cap I_s(x^+),
\end{equation}
\begin{equation}
I_{se}(x,x^+)=  I_s(x) \cap I_e(x^+),
\end{equation}
\begin{equation}
I_{ee}(x,x^+)=  I_e(x) \cap I_e(x^+),
\end{equation}\begin{equation}
I_{es}(x,x^+)= I_e(x) \cap I_s(x^+),
\end{equation}
where $I_{ss}(x,x^+)$ are the agents that remain stationary,  $I_{se}(x,x^+)$ are the ones starting to experiment, $I_{e}(x)$ are  the experimenting agents that have not completed moving along their experiment paths, and $I_{es}(x,x^+)$ are the agents that have completed traversing their experiment paths and choose between their candidate actions. 


The agents in $I_{es}(x,x^+)$ can be further partitioned as the ones choosing their first candidate action and the ones that choose their second candidate action, i.e.
\begin{equation}\label{Id1}
I_{es}^1(x,x^+)= \{i \in I_{es}(x,x^+) \mid a_i^{+}=a_i^1  \},
\end{equation}
\begin{equation}\label{Id1}
I_{es}^2(x,x^+)= \{i \in I_{es}(x,x^+) \mid a_i^{+}=a_i^2  \}.
\end{equation}

Note that the agents in $I_{es}(x,x^+)$ do not necessarily choose the action resulting in the higher estimated utility. For each $i \in I$, let $\hat{U}_i^*$=max\{$\hat{U}_i^1,\hat{U}_i^2$\}. Then, the amount of estimated utility that is denied in the transition $x \rightarrow x^+$ is given as
\begin{equation}
\label{denied}
\Delta_i(x_i,x_i^+)=  \left\{\begin{array}{ll} \hat{U}_i^*-\hat{U}_i^1&\mbox{if $i \in I_{es}^1(x,x^+)$}, \\ \hat{U}_i^*-\hat{U}_i^2&\mbox{if $i \in I_{es}^2(x,x^+) $}, \\ 0&\mbox{otherwise. }\end{array}\right.
\end{equation}

Next, we show that the CFCM algorithm induces a regular perturbed Markov chain, where the resistance of any feasible transition depends on the estimated utilities denied by the agents becoming stationary and the number of agents starting new experiments. 
 
 

\begin{lem}
\label{RPMC}
Let  $\mathcal{G}=(V,E)$ be connected graph. If all agents employ the CFCM algorithm, then a regular perturbed Markov chain is induced over $\mathcal{X}$, and the resistance of any feasible transition, $x \rightarrow  x^+$, is  
\begin{equation}
\label{resist}
R(x, x^+) = r|I_{se}(x, x^+)| + \sum_{i \in I_{es}(x,x^+)} \Delta_i(x_i, x^+_i).
\end{equation}
\end{lem}
\begin{proof}
Let $P_\epsilon$ denote the transition matrix of the Markov chain induced by the CFCM algorithm. For $\epsilon >0$, any all-stationary state can be reached from any other all-stationary state through a sequence of experiments, given $\mathcal{G}=(V,E)$ is connected. Furthermore, any state that is not all-stationary lies on a feasible path between two all-stationary states. Hence, $P_{\epsilon}$ is irreducible. Furthermore, since the stationary agents remain stationary with probability $1-\epsilon^r$, aperiodicity immediately follows from the resulting self-loops at all-stationary states.

 The probability any feasible transition from $x$ to $x^+$, given in $P_{\epsilon}$, is the joint probability of state transitions of individual agents. Note that for any agent,  $i \in I_{ee}(x)$, the transition from $x_i$ to $x_i^+$ does not have any randomness. Hence, the probability of transition  from $x$ to $x^+$ is 
\begin{multline}
\label{probs}
P_{\epsilon}(x, x^+) =   \Pr[I_{es}^1(x, x^+)]\Pr[I_{es}^2(x, x^+)]\Pr[I_{ss}(x, x^+)] \\ \Pr[I_{se}(x, x^+)],
\end{multline}
where each term on the right side of (\ref{probs}) denote the joint probability of state transitions for the agents in the corresponding subset, and they are given as
\begin{equation}
\label{PID1}
\Pr[I_{es}^1(x, x^+)]=\prod_{i\in I_{es}^1(x, x^+) } \frac{\epsilon^{-\hat{U}_i^1}}{\epsilon^{-\hat{U}_i^1}+\epsilon^{-\hat{U}_i^2}},
\end{equation}
\begin{equation}
\label{PID2}
\Pr[I_{es}^2(x, x^+)]=\prod_{i\in I_{es}^2(x, x^+) } \frac{\epsilon^{-\hat{U}_i^2}}{\epsilon^{-\hat{U}_i^1}+\epsilon^{-\hat{U}_i^2}},
\end{equation}
\begin{equation}
\label{PIS}
\Pr[I_{ss}(x, x^+)]=\prod_{i\in I_{ss}(x, x^+) } (1-\epsilon^r), 
\end{equation}
\begin{equation}
\label{PIE}
\Pr[I_{se}(x, x^+)]=\prod_{i\in I_{se}(x, x^+) } \frac{\epsilon^r}{|A_i^c(a_i^1)|} \Pr[S_i^{+} ; a_i^1, a_i^2], 
\end{equation}
where $\Pr[S_i^{+} ; a_i^1, a_i^2]$ is the probability of having $S_i^+$ as the experiment path for an agent comparing $a_i^1$ and $a_i^2$. $\Pr[S_i^{+} ; a_i^1, a_i^2]$ depends on the function $\mathcal{E}(a_i^1, a_i^2)$, and it is independent of $\epsilon$.
Plugging (\ref{PID1})-(\ref{PIE}) into (\ref{probs}),  one can verify that the resistance $R(x,x^+)$ given in (\ref{resist}) satisfies
 \begin{equation}
\label{res}
0 < \lim_ {\epsilon \rightarrow 0^+} \frac{P_{\epsilon}(x, x^+)}{\epsilon^{R(x,x^+)}} < \infty.
\end{equation}
\end{proof}

Since the CFCM algorithm induces a regular perturbed Markov chain, the stochastically stable states are the recurrent states of the unperturbed chain with the minimum stochastic potential, as given in Lemma \ref{stochrem}. Note that if $\epsilon =0$, then no agent starts an experiment. In that case, the set of recurrent states, $\mathcal{X}^0_R$, consists of the all-stationary states. All the other states, where at least one agent is experimenting, form the set of transient states, $\mathcal{X}^0_T$, i.e.

\begin{equation}
\label{recurstat}
\mathcal{X}^0_R = \{ x \mid I_s(x)=I\},
\end{equation}
\begin{equation}
\label{transtat}
\mathcal{X}^0_T = \mathcal{X}\setminus \mathcal{X}^0_R.
\end{equation}

The stochastic potentials of the states in $\mathcal{X}^0_R$ are determined by the resistances of the feasible transitions.  Note that the parameter $r$ in the CFCM algorithm has a direct influence on the resistances as given in (\ref{resist}).  We will show that, for any connected graph $\mathcal{G}$, if $r$ is sufficiently large, then the states in $\mathcal{X}^0_R$ with the minimum stochastic potential are the coverage maximizers. To provide a sufficient value of $r$, first we relate the structure of the graph to the maximum amount of estimated utility that can be denied by an agent in any feasible transition under the CFCM algorithm.


\begin{lem}
\label{gslem}
Let all agents follow the CFCM algorithm to cover a connected graph, $\mathcal{G}=(V,E)$, and let $\nu(\mathcal{G})$ be 
\begin{equation}
\label{graphdstar}
\nu(\mathcal{G}) = \max _{(v,v') \in E}|\mathcal{N}_{v}^\delta \setminus \mathcal{N}_{v'}^\delta|.
\end{equation}
Then, for any feasible transition $x\rightarrow x^+$, 
\begin{equation}
\label{gs1}
\nu(\mathcal{G}) \geq \max _{i \in I}\Delta_i(x_i,x_i^+).
\end{equation}
\end{lem}
\begin{proof}
Let $x\rightarrow x^+$ be a feasible transition. For any $i \in I_{s}(x)$, $\hat{U}_i^1=\hat{U}_i^2=0$. On the other hand, for any $i \in I_{e}(x)$, the sampled partial utilities from the nodes $\mathcal{N}_{a_i^1}^{\delta} \cap \mathcal{N}_{a_i^2}^{\delta}$, contribute equally to both $\hat{U}_i^1$ and $\hat{U}_i^2$. Hence,

\begin{equation}
\label{gs2}
\mbox{max} \{|\mathcal{N}_{a_i^1}^{\delta} \setminus \mathcal{N}_{a_i^2}^{\delta}|, |\mathcal{N}_{a_i^2}^{\delta} \setminus \mathcal{N}_{a_i^1}^{\delta}|\} \geq |\hat{U}_i^1-\hat{U}_i^2|, \; \forall i \in I.
\end{equation}
In light of (\ref{gs2}) and (\ref{denied}), 
\begin{equation}
\label{gs3}
\mbox{max} \{|\mathcal{N}_{a_i^1}^{\delta} \setminus \mathcal{N}_{a_i^2}^{\delta}|, |\mathcal{N}_{a_i^2}^{\delta} \setminus \mathcal{N}_{a_i^1}^{\delta}|\} \geq \Delta_i(x_i,x_i^+), \; \forall i \in I.
\end{equation}
Since $(a_i^1, a_i^2) \in E$ for any $i \in I_e(x)$, (\ref{graphdstar}) implies
\begin{equation}
\label{gs4}
\nu(\mathcal{G}) \geq \mbox{max} \{|\mathcal{N}_{a_i^1}^{\delta} \setminus \mathcal{N}_{a_i^2}^{\delta}|, |\mathcal{N}_{a_i^2}^{\delta} \setminus \mathcal{N}_{a_i^1}^{\delta}|\}, \; \forall i \in I.
\end{equation}
Finally, (\ref{gs3}) and (\ref{gs4}) together imply (\ref{gs1}).

\end{proof}


Next, we show that $r> \nu(\mathcal{G})$ is a sufficient condition to ensure that the paths between the states in $\mathcal{X}_R^0$ on a minimum resistance tree consist of unilateral experimentations. 

\begin{definition} \emph{(Unilateral Experimentation Path)}: A feasible sequence of states, $\mathcal{P}=\{ x^1, x^2, \hdots x^n \}$, is a unilateral experimentation path if $x^1,x^n \in \mathcal{X}_R^0$, $x^2, \hdots, x^{n-1} \in \mathcal{X}_T^0$  and for all $1\leq p \leq n-1$
\begin{equation}
\label{unit}
 |I_{se}(x^p,x^{p+1})|=  \left\{\begin{array}{ll} 1&\mbox{if $p=1$}, \\ 0&\mbox{otherwise. }\end{array}\right.
 \end{equation}
\end{definition}
\begin{lem}
\label{unilateral1}
Let $\mathcal{T}^*$ be a minimum resistance tree, and let $x\rightarrow x^+ \in \mathcal{T}^*$. If $x \in \mathcal{X}^0_R$, then $|I_{se}(x,x^+)|=1$. 
\end{lem}
\begin{proof} Since $x \in \mathcal{X}^0_R$, $|I_{se}(x,x^+)|>0$, as otherwise, $x^+=x$ and $x\rightarrow x^+$ cannot be contained in a tree. Assume that $|I_{se}(x,x^+)|>1$. Then, choose an arbitrary $i \in I_{se}(x,x^+)$ to define an $\tilde{x}^+ \neq x$ as
\begin{equation}
\tilde{x}_j^+ = \left\{\begin{array}{ll} x^+_j &\mbox{ if $j \neq i$}, \\ x_i &\mbox{ otherwise. }\end{array}\right.
\end{equation} 
Note that $x\rightarrow \tilde{x}^+$ is a feasible transition, and $R(x, \tilde{x}^+) = R(x, x^+)-r(|I_{se}(x,x^+)|-1)$. Replacing $x\rightarrow x^+ $ with $x\rightarrow \tilde{x}^+$ would give an alternative tree with a smaller resistance, which contradicts with $\mathcal{T}$ being a minimum resistance tree.
\end{proof}

\begin{lem}
\label{unilateral2}
Let $\mathcal{T}^*$ be a minimum resistance tree, and let $x\rightarrow x^+ \in \mathcal{T}^*$. If $x \in \mathcal{X}^0_T$ and $r>\nu(\mathcal{G})$,  then we have $|I_{se}(x,x^+)| < |I_e(x)|$.
\end{lem}
\begin{proof} Since $x \in \mathcal{X}^0_T$,  $I_{se}(x,\tilde{x}^+)= \emptyset$ doesn't imply $\tilde{x}^+=x$.  Hence, there exists an $\tilde{x}^+ \neq x$ such that $x\rightarrow \tilde{x}^+$ is feasible and $I_{se}(x,\tilde{x}^+)= \emptyset$. For any such $\tilde{x}^+$, we have 
\begin{equation}
\label{altern1}
R(x, \tilde{x}^+)-R(x, x^+)  \leq - r |I_{se}(x,x^+)|+  |I_{es}(x,\tilde{x}^+)| \nu(\mathcal{G}).
\end{equation} 
Note that $|I_{es}(x,\tilde{x}^+)| \leq |I_{e}(x)|$. Hence, given $r>\nu(\mathcal{G})$, the right side of (\ref{altern1}) is negative for any $|I_{se}(x,x^+)| \geq |I_e(x)|$. In that case, replacing $x\rightarrow x^+ $ with $x\rightarrow \tilde{x}^+$ would give an alternative tree with a smaller resistance, which contradicts with $\mathcal{T}$ being a minimum resistance tree. Consequently, $|I_{se}(x,x^+)| < |I_e(x)|$.
 \end{proof}

\begin{lem}
\label{unilateral}
Let $r>\nu(\mathcal{G})$, and let $\mathcal{P}=\{ x^1, x^2, \hdots x^n \}$ be a sequence of states, where $x^1,x^n \in \mathcal{X}^0_R$ and $x^2, \hdots, x^{n-1} \in \mathcal{X}^0_T$. If $\mathcal{P} \in \mathcal{T}$ for some minimum resistance tree $\mathcal{T}$, then $\mathcal{P}$ is a unilateral experimentation path.
\end{lem}
\begin{proof} Since $x^1\in \mathcal{X}^0_R$, from Lemma \ref{unilateral1}, we have $|I_{se}(x^1,x^2)|=1$ leading to $|I_e(x^2)|=1$. Furthermore, for $r>\nu(\mathcal{G})$, from Lemma \ref{unilateral2}, we have $|I_{se}(x^2,x^3)|=0$. Hence, we have $|I_e(x^3)| \leq 1$. Using Lemma \ref{unilateral2} recursively along $\mathcal{P}$ we obtain
\begin{equation}
\label{unitexpprop}
|I_{se}(x^p,x^{p+1})|= \left\{\begin{array}{ll} 1 &\mbox{ if $p=1$}, \\ 0 &\mbox{ otherwise. }\end{array}\right.
\end{equation}
Hence, $\mathcal{P}$ is a unilateral experimentation path.
\end{proof}

\begin{lem}
\label{singleres}
If $\mathcal{P}=\{ x^1, x^2, \hdots x^n \}$ be a unilateral experimentation path, then 
\begin{equation}
\label{resist1}
R(\mathcal{P})= \sum_{p=1}^{n-1}R(x^p,x^{p+1})=r+\max\{\phi(x^n),\phi(x^1)\}-\phi(x^n).
\end{equation} 
\end{lem}
\begin{proof} Since $\mathcal{P}=\{ x^1, x^2, \hdots x^n \}$ be a unilateral experimentation path, for $x^p, x^{p+1} \in \mathcal{X}_T^0$, we have 
\begin{equation}
\label{singleres1}
|I_{se}(x^p,x^{p+1})|= |I_{es}(x^p,x^{p+1})|=0.
\end{equation}
Hence, such transitions have zero resistance, resulting in 
\begin{equation}
\label{singleres2}
R(\mathcal{P})= R(x^1,x^2)+R(x^{n-1},x^n).
\end{equation}
Note that, since $x^1 \in \mathcal{X}_R^0$ and $\mathcal{P}$ is a unilateral experimentation path, we have $R(x^{1},x^2)=r$ and $R(x^{n-1},x^n)=\Delta_i(x_i^{n-1},x_i^n)$,
where $i\in I$ is the unique experimenting agent. Since all the other agents are stationary, i.e. $a_{-i}$ is constant along $\mathcal{P}$, the estimated utilities satisfy 
\begin{equation}
\label{singleres3}
(\hat{U}^1_i)^{n-1}= \sum_{v \in \mathcal{N}^{\delta}_{a^1_i}}u(v, a_{-i})= U_i(a^1_i, a_{-i}),
\end{equation}
\begin{equation}
\label{singleres4}
(\hat{U}^2_i)^{n-1}= \sum_{v \in \mathcal{N}^{\delta}_{a^2_i}}u(v, a_{-i})=U_i(a^2_i, a_{-i}).
\end{equation}
Plugging (\ref{singleres3}) and (\ref{singleres4}) into (\ref{denied}) we obtain
\begin{equation}
\label{singleres5}
\Delta_i(x_i^{n-1},x_i^n)= \max\{U_i(x^n),U_i(x^1)\}-U_i(x^n).
\end{equation}
Since $\Gamma_{\text{DGC}}$ is a potential game, from (\ref{singleres5}) we obtain
\begin{equation}
\label{singleres6}
\Delta_i(x_i^{n-1},x_i^n)= \max\{\phi(x^n),\phi(x^1)\}-\phi(x^n).
\end{equation}
 \end{proof}


\begin{lem}
\label{maxpot}
Let $r>\nu(\mathcal{G})$, and let $\mathcal{T}^*_x$ and $\mathcal{T}^*_{x'}$ be minimum resistance trees rooted at some $x, x'  \in \mathcal{X}_R^0$. Then,
\begin{equation}
\label{maxpoteq}
R(\mathcal{T}_x^*) \leq R(\mathcal{T}_{x'}^*) \Rightarrow \phi(x)\geq \phi(x').
\end{equation}
\end{lem}
\begin{proof} 
For $r>\nu(\mathcal{G})$, in light of Lemma \ref{unilateral}, the paths between the states in $\mathcal{X}_R^0$ on a minimum resistance tree consist of unilateral experimentations.  Let $x^0_R \in \mathcal{X}_R^0$, and let  $\mathcal{T}^*_{x_R^0}$ be a minimum resistance tree rooted at $x^0_R$. Let $x^n_R \in \mathcal{X}_R^0$ be a state such that $R(\mathcal{T}^*_{x_R^0}) \leq R(\mathcal{T}^*_{x_R^n})$ and the unique path, $\mathcal{P}\in \mathcal{T}^*_{x_R^0}$, from $x^n_R$ to $x^0_R$ consists of $n$ unilateral experimentations, i.e.
\begin{equation}
\label{multP}
R(\mathcal{P}) = \sum_{k=1}^{n}R(\mathcal{P}_k),
\end{equation}
where $\mathcal{P}_k$ is the unilateral experimentation starting at $x_R^{n-k+1}$ and ending at $x_R^{n-k}$. Note that, for each such $\mathcal{P}_k$, there exists a feasible unilateral experimentation path $\mathcal{P}_k'$ in the reversed direction, starting at $x_R^{n-k}$ and ending at  $x_R^{n-k+1}$. Replacing each $\mathcal{P}_k$ with $\mathcal{P}_k'$, one can construct a tree rooted at  $\mathcal{T}_{x_R^n}$. Note that the resistances of these trees satisfy
\begin{eqnarray}
\label{multP2}
R(\mathcal{T}_{x_R^n}) - R(\mathcal{T}_{x_R^0}^*) & =& \sum_{k=1}^{n}(R(\mathcal{P}_k')-R(\mathcal{P}_k)) \nonumber \\
&=& \sum_{k=1}^{n} (\phi(x_R^{n-k}) - \phi(x_R^{n-k+1})) \nonumber  \\
&=& \phi(x_R^{0}) - \phi(x_R^{n}). 
\end{eqnarray}
Note that by definition $R(\mathcal{T}^*_{x_R^n}) \leq R(\mathcal{T}_{x_R^n})$. Hence, if $R(\mathcal{T}^*_{x_R^0}) \leq R(\mathcal{T}^*_{x_R^n})$, then $R(\mathcal{T}^*_{x_R^0}) \leq R(\mathcal{T}_{x_R^n})$ for any $\mathcal{T}_{x_R^n}$. Plugging this into (\ref{multP2}), we obtain $\phi(x_R^{0}) \geq \phi(x_R^{n})$

\end{proof}

\begin{theorem}
\label{maintheo}
Let  $\mathcal{G}=(V,E)$ be connected graph. Let all agents follow the $CFCM$ algorithm with $r> \nu(\mathcal{G})$, and let $x$ be a stochastically stable state of the resulting Markov chain. Then, $x \in \mathcal{X}_R^0$ and 
\begin{equation}
|V_c(x)| \geq |V_c(x')|, \; \forall x'\in \mathcal{X}_R^0.
\label{limdistCFCM}
\end{equation}

\end{theorem}
\begin{proof} 
Let $x$ be a stochastically stable state. Due to Lemma \ref{stochrem}, $x \in \mathcal{X}_R^0$ and $R(\mathcal{T}^*_x) \leq R(\mathcal{T}^*_{x'})$ for all $x' \in \mathcal{X}_R^0$. In light of Lemma \ref{maxpot}, if $r>\nu(\mathcal{G})$, then $R(\mathcal{T}^*_x) \leq R(\mathcal{T}^*_{x'})$ implies $\phi(x) \geq \phi(x')$ for all $x' \in \mathcal{X}_R^0$. As such, (\ref{limdistCFCM}) is satisfied since $\phi(x)=|V_c(x)|$.
\end{proof}

Theorem \ref{maintheo} indicates that if all agents follow the CFCM algorithm with sufficiently large $r$, then the stochastically stable states are all-stationary states maximizing the number of covered nodes. As such, the agents asymptotically maintain maximum coverage with an arbitrarily high probability for arbitrarily small values of the noise parameter $\epsilon$.

\section{Simulation Results}
\label{sims}
In this section, some simulation results are presented to demonstrate the performance of the proposed method. In the simulation, a group of $13$ agents are initially placed at an arbitrary node of a connected random geometric graph. Each agent has a sensing range $\delta=1$. The graph consists of $50$ nodes and $78$ edges, and it has $\nu(\mathcal{G})=4$. Note that  $r>\nu(\mathcal{G})$ is a sufficient condition for the stochastic stability of potential maximizers due to the sufficiently high resistance of simultaneous experiments as given in Lemma \ref{unilateral}. However, $r>\nu(\mathcal{G})$ may not be necessary in many cases since simultaneously updating agents do not necessarily influence the utility estimations of each other, especially when they are sufficiently far from each other. In this simulation, the agents follow the CFCM algorithm with $\epsilon=0.015$ and $r=1.5$.

All the agents are initially stationary at the same position on the graph. The number of covered nodes throughout a period of 200000 time steps is shown in Fig. \ref{uncht_cf}, whereas the configuration of the agents on the graph at some instants are provided in Fig. \ref{covtime_cf}. As depicted in Fig. \ref{uncht_cf}, after a sufficient amount of time, the agents maintain complete coverage with a very high probability. For $150000 \leq t \leq 200000$, the average number of covered nodes at each time step is computed as $49.7$.  



\begin{figure}[htb]
\centering
\includegraphics[trim =45mm 0mm 10mm 10mm, clip,scale=0.27]{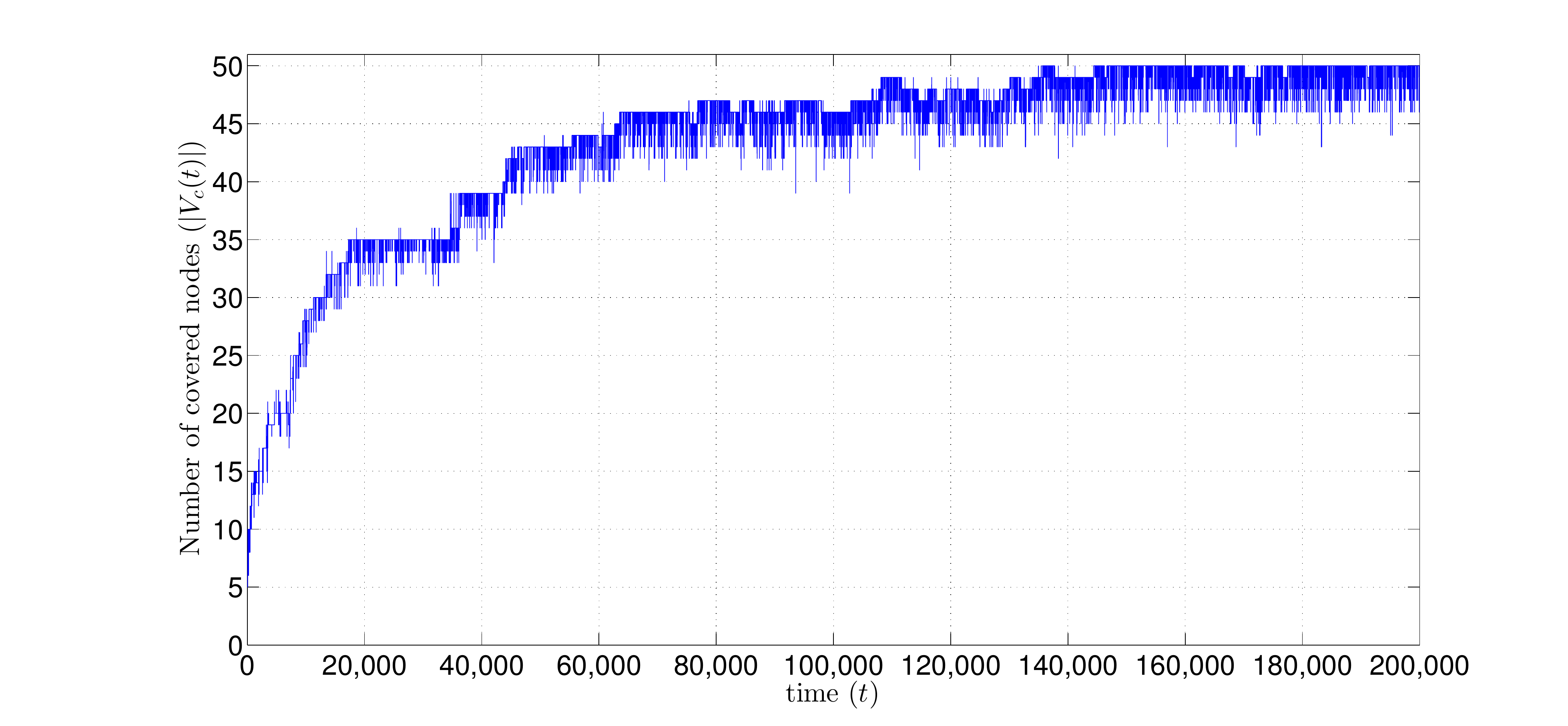}
\caption{The number of covered nodes as a function of time (CFCM).  }
\label{uncht_cf}
\end{figure}
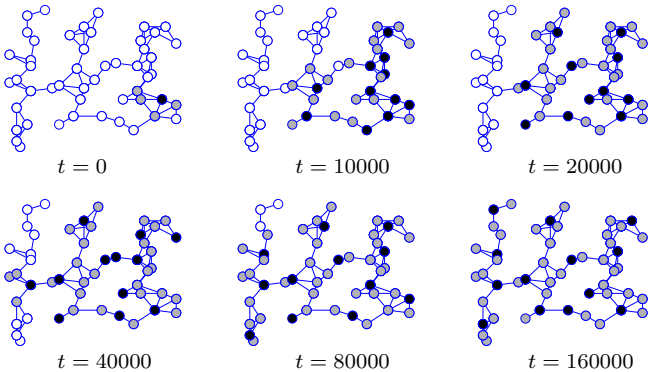
\begin{figure}[htb]
\centering
\includegraphics[trim =40mm 5mm 0mm 10mm, clip,scale=0.29]{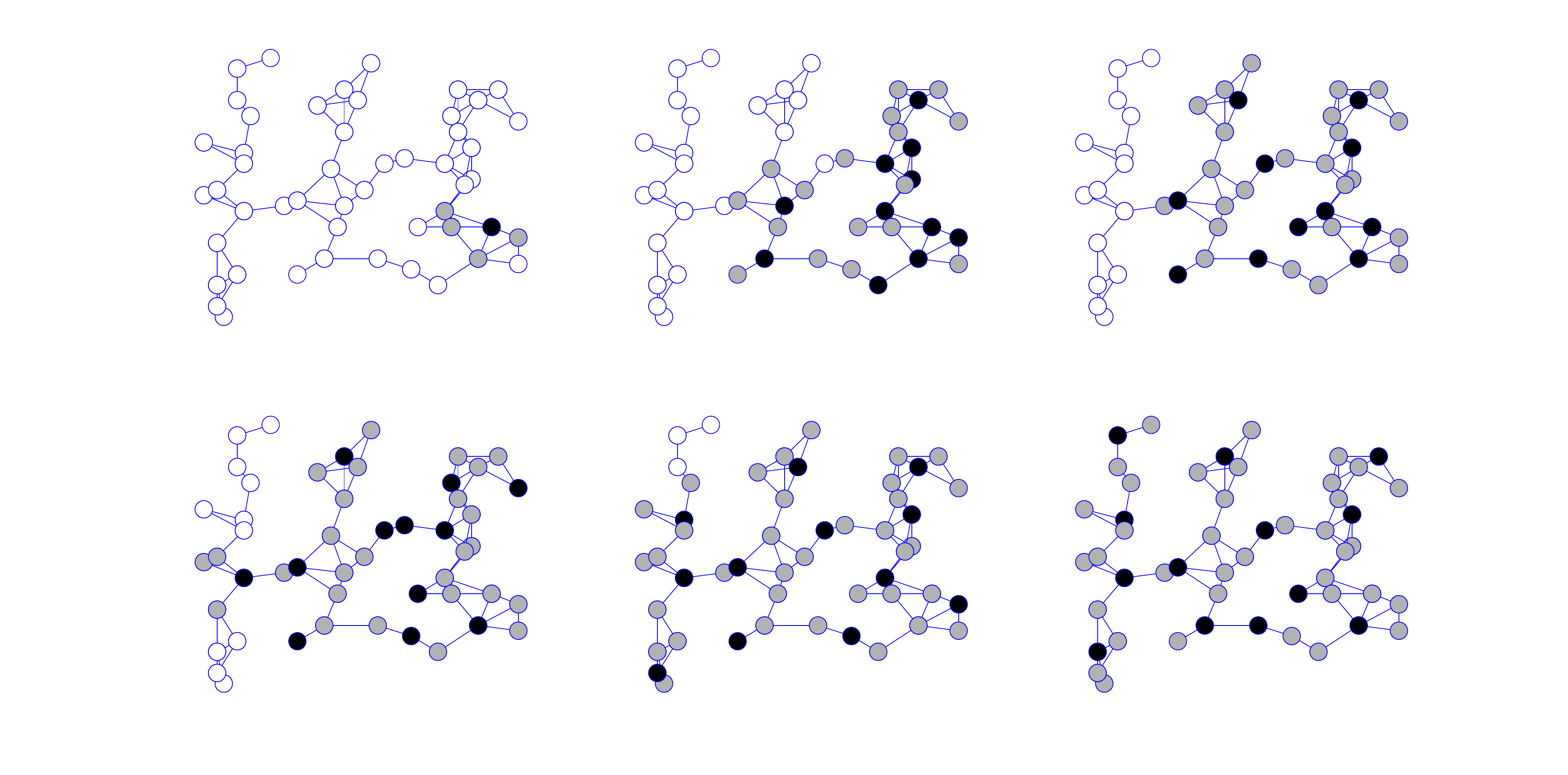}
\caption{The configuration of 13 agents on the graph at some instants of the simulation (CFCM). The nodes occupied by at least one agent are black, the nodes covered by at least one agent are gray, and the nodes that are not covered are white.  }
\label{covtime_cf}
\begin{picture}(0,0)
\footnotesize
\put(-10,123){$t=10000$}
\put(-100,123){$t=0$}
\put(78,123){$t=20000$}
\put(-100,49){$t=40000$}
\put(-10,49){$t=80000$}
\put(78,49){$t=160000$}
\end{picture}
\end{figure}
\vskip4ex
In order to compare the performance with a setting that allows for communications, we also  present a simulation of the same scenario with BLLL. The agents start at the same initial condition as the previous simulation, and BLLL is executed with the same noise parameter $\epsilon=0.015$. The number of covered nodes throughout a period of 10000 time steps is shown in Fig. \ref{cov_ones}, whereas the configuration of the agents on the graph at some instants are provided in Fig. \ref{traj_ones}. As illustrated in Fig. \ref{cov_ones}, after a sufficient amount of time, the agents maintain complete coverage with a very high probability. For $7500 \leq t \leq 10000$, the average number of covered nodes at each time step is computed as $49.76$.


\begin{figure}[htb]
\centering
\includegraphics[trim =45mm 0mm 10mm 10mm, clip,scale=0.3]{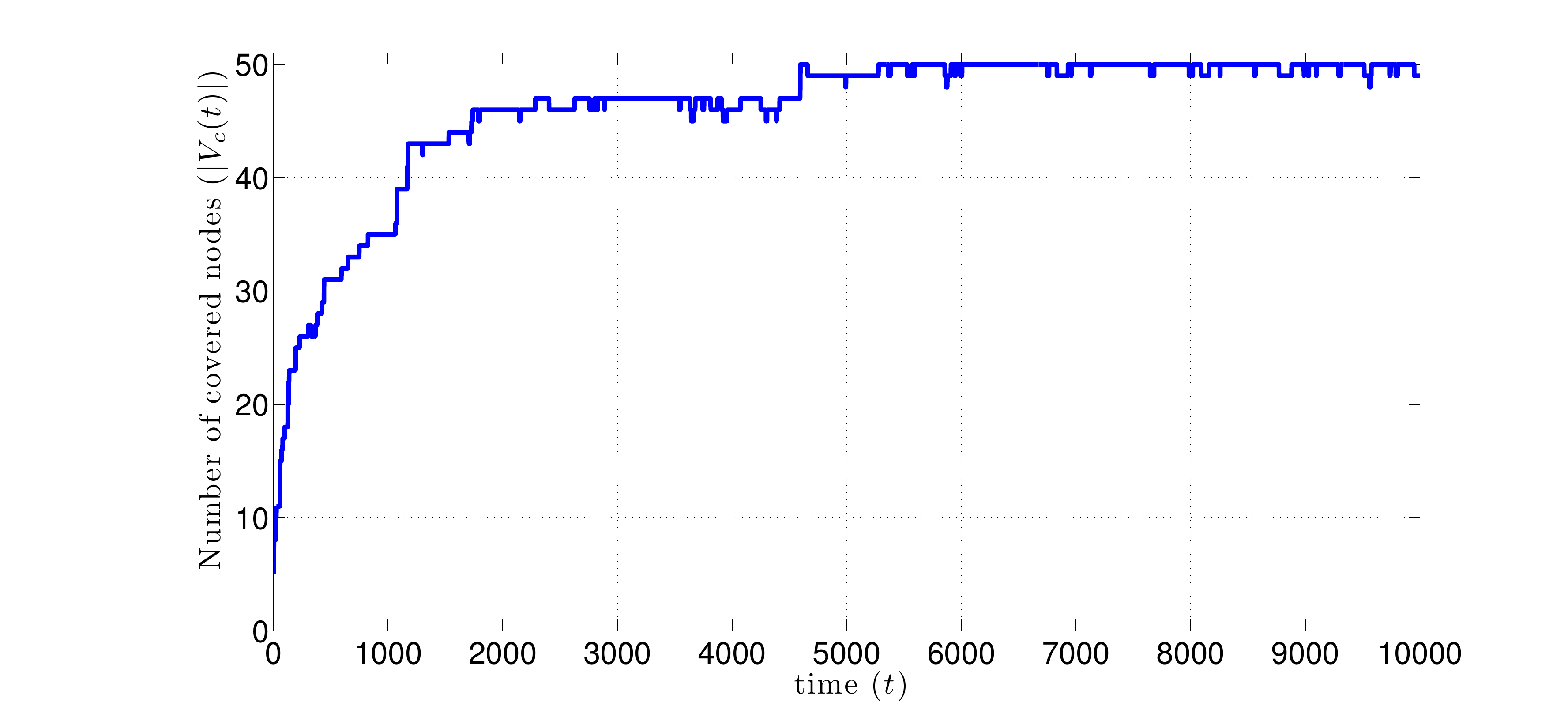}
\caption{The number of covered nodes as a function of time (BLLL). }
\label{cov_ones}
\end{figure}

\begin{figure}[htb]
\centering
\includegraphics[trim =40mm 5mm 0mm 10mm, clip,scale=0.33]{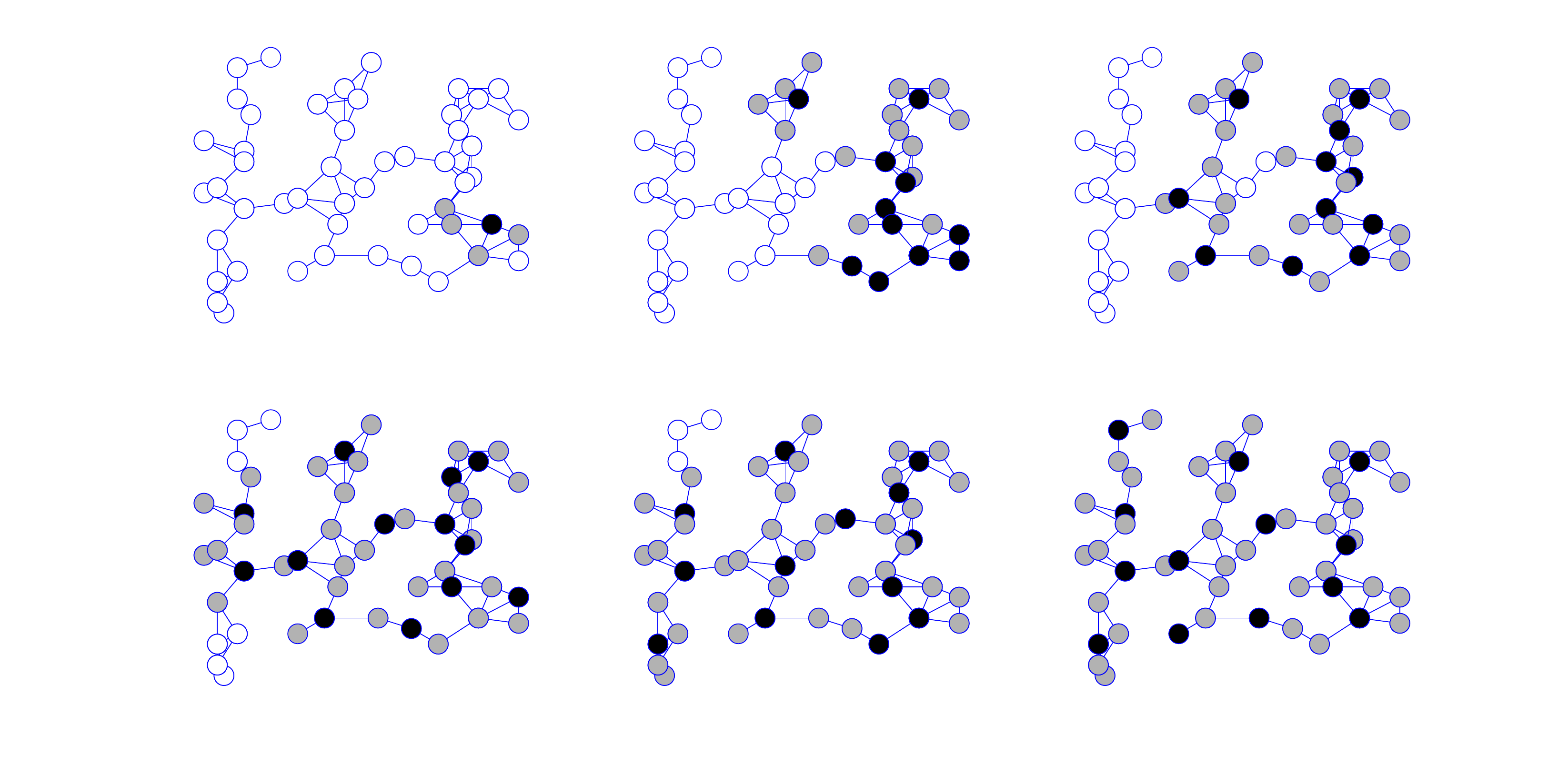}
\caption{The configuration of 13 agents on the graph at some instants of the simulation (BLLL). The nodes occupied by at least one agent are black, the nodes covered by at least one agent are gray, and the nodes that are not covered are white.  }
\label{traj_ones}
\begin{picture}(0,0)
\footnotesize
\put(-15,122){$t=350$}
\put(-102,122){$t=0$}
\put(75,122){$t=700$}
\put(-102,47){$t=1400$}
\put(-15,47){$t=2800$}
\put(75,47){$t=5600$}
\end{picture}
\end{figure}
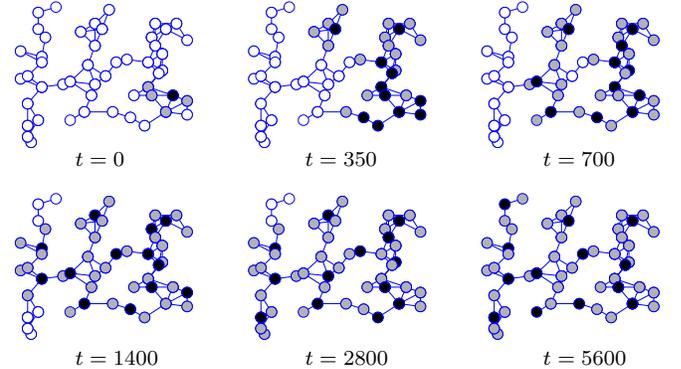

Through the comparison of Figs. \ref{uncht_cf} and \ref{covtime_cf}. to Figs. \ref{cov_ones} and \ref{traj_ones}, it is seen that both algorithms drive the agents to some global optima in a similar fashion. However, when the agents are allowed to communicate, they can maximize the coverage much faster, as one might expect. Despite the slower convergence to the limiting distribution, the main advantage of the CFCM algorithm is that the agents do not need to know their actual utilities whose computation requires some communications in the DGC problem. As such, CFCM can be employed to optimally distribute some mobile security resources on networks, even in scenarios that do not allow for such explicit communications. 
\section{Conclusion}
\label{conclusion}

In this paper, a game theoretic approach was proposed for distributed coverage of networked systems by mobile agents with local capabilities. We considered a distributed graph coverage (DGC) problem, where the network is modeled as an undirected graph, and the agents are located on some nodes of the graph. Each agent can sense the graph structure and the presence of the other agents within its $\delta$-neighborhood, where $\delta$ is the sensing range. Any node of the graph is covered if it is within the sensing range of at least one agent. The agents move locally on the graph, and they aim to maximize the number of covered nodes. We studied this problem particularly for agents with no explicit communications among themselves. 

A game theoretic formulation of the DGC problem was obtained by designing a potential game, $\Gamma_{\text{DGC}}$. In $\Gamma_{\text{DGC}}$, the action of each agent is defined as its position on the graph, and the utility of each agent is equal to the number of nodes covered only by itself. It was shown that $\Gamma_{\text{DGC}}$ can be paired with a learning algorithm such as BLLL to maximize the coverage. However, such learning algorithms require the agents to measure their current utilities. In $\Gamma_{\text{DGC}}$, the actual utilities can not be computed without explicit communications since the agents with overlapping coverage are not necessarily within the sensing range of each other. In order to address this issue, we presented a communication-free learning algorithm, namely the CFCM. In CFCM, the agents follow a noisy best-response policy based on the estimated utilities gathered by moving around their current positions. The algorithm has a noise parameter,  $\epsilon \in \Re^+$, and a second parameter,  $r \in \Re^+$, that sets the likelihood of remaining stationary. We showed that the CFCM algorithm induces a regular perturbed Markov chain and the stochastically stable states are the coverage maximizers for sufficiently large values of $r$. A sufficient value of $r$ was derived from the topology of the graph. Some simulation results were also presented to demonstrate that the CFCM algorithm achieves optimal coverage.



\bibliography{MyReferences}

\begin{thebibliography}{10}

\bibitem{Goddard05}
W.~Goddard, S.~M. Hedetniemi, and S.~T. Hedetniemi, ``Eternal security in
  graphs,'' {\em J. Combin. Math. Combin. Comput}, vol.~52, pp.~169--180, 2005.

\bibitem{Du03}
T.~C. Du, E.~Y. Li, and A.-P. Chang, ``Mobile agents in distributed network
  management,'' {\em Communications of the ACM}, vol.~46, no.~7, pp.~127--132,
  2003.

\bibitem{Berbeglia10}
G.~Berbeglia, J.-F. Cordeau, and G.~Laporte, ``Dynamic pickup and delivery
  problems,'' {\em European Journal of Operational Research}, vol.~202, no.~1,
  pp.~8--15, 2010.

\bibitem{Reese06}
J.~Reese, ``Solution methods for the p-median problem: An annotated
  bibliography,'' {\em Networks}, vol.~48, no.~3, pp.~125--142, 2006.

\bibitem{Megiddo83}
N.~Megiddo, E.~Zemel, and S.~L. Hakimi, ``The maximum coverage location
  problem,'' {\em SIAM Journal on Algebraic Discrete Methods}, vol.~4, no.~2,
  pp.~253--261, 1983.

\bibitem{Khuller99}
S.~Khuller, A.~Moss, and J.~S. Naor, ``The budgeted maximum coverage problem,''
  {\em Information Processing Letters}, vol.~70, no.~1, pp.~39--45, 1999.

\bibitem{Owen98}
S.~H. Owen and M.~S. Daskin, ``Strategic facility location: A review,'' {\em
  European Journal of Operational Research}, vol.~111, no.~3, pp.~423--447,
  1998.

\bibitem{Caprara2000}
A.~Caprara, P.~Toth, and M.~Fischetti, ``Algorithms for the set covering
  problem,'' {\em Annals of Operations Research}, vol.~98, no.~1-4,
  pp.~353--371, 2000.

\bibitem{Howard02}
A.~Howard, M.~J. Matari{\'c}, and G.~S. Sukhatme, ``Mobile sensor network
  deployment using potential fields: A distributed, scalable solution to the
  area coverage problem,'' in {\em Distributed Autonomous Robotic Systems 5},
  pp.~299--308, Springer, 2002.

\bibitem{Schwager09}
M.~Schwager, D.~Rus, and J.-J. Slotine, ``Decentralized, adaptive coverage
  control for networked robots,'' {\em International Journal of Robotics
  Research}, vol.~28, no.~3, pp.~357--375, 2009.

\bibitem{Poduri04}
S.~Poduri and G.~S. Sukhatme, ``Constrained coverage for mobile sensor
  networks,'' in {\em IEEE International Conference on Robotics and
  Automation}, pp.~165--171, 2004.

\bibitem{Cortes04}
J.~Cort\'{e}s, S.~Mart\'{i}nez, T.~Karatas, and F.~Bullo, ``Coverage control
  for mobile sensing networks,'' {\em IEEE Transactions on Robotics and
  Automation}, vol.~20, no.~2, pp.~243--255, 2004.

\bibitem{Lloyd82}
S.~Lloyd, ``Least squares quantization in pcm,'' {\em IEEE Transactions on
  Information Theory}, vol.~28, no.~2, pp.~129--137, 1982.

\bibitem{Cortes05}
J.~Cortes, S.~Martinez, and F.~Bullo, ``Spatially-distributed coverage
  optimization and control with limited-range interactions,'' {\em ESAIM:
  Control, Optimisation and Calculus of Variations}, vol.~11, no.~4,
  pp.~691--719, 2005.

\bibitem{Kwok10}
A.~Kwok and S.~Mart{\'\i}nez, ``Deployment algorithms for a power-constrained
  mobile sensor network,'' {\em International Journal of Robust and Nonlinear
  Control}, vol.~20, no.~7, pp.~745--763, 2010.

\bibitem{Pimenta08}
L.~Pimenta, V.~Kumar, R.~C. Mesquita, and G.~Pereira, ``Sensing and coverage
  for a network of heterogeneous robots,'' in {\em IEEE Conference on Decision
  and Control}, pp.~3947--3952, 2008.

\bibitem{Durham09}
J.~W. Durham, R.~Carli, P.~Frasca, and F.~Bullo, ``Discrete partitioning and
  coverage control with gossip communication,'' in {\em ASME Dynamic Systems
  and Control Conference}, pp.~225--232, 2009.

\bibitem{Yun12}
S.~Yun and D.~Rus, ``Distributed coverage with mobile robots on a graph:
  Locational optimization,'' in {\em IEEE International Conference on Robotics
  and Automation}, pp.~634--641, 2012.

\bibitem{Zhu13}
M.~Zhu and S.~Mart\'{\i}nez, ``Distributed coverage games for energy-aware
  mobile sensor networks,'' {\em SIAM Journal on Control and Optimization},
  vol.~51, no.~1, pp.~1--27, 2013.

\bibitem{Yasin13NecSys}
A.~Y. Yaz{\i}c{\i}o\u{g}lu, M.~Egerstedt, and J.~S. Shamma, ``A game theoretic
  approach to distributed coverage of graphs by heterogeneous mobile agents,''
  in {\em IFAC Workshop on Distributed Estimation and Control in Networked
  Systems}, pp.~309--315, 2013.

\bibitem{Arslan07}
G.~Arslan, J.~Marden, and J.~S. Shamma, ``Autonomous vehicle-target assignment:
  a game theoretical formulation,'' {\em ASME Journal of Dynamic Systems,
  Measurement, and Control}, pp.~584--596, 2007.

\bibitem{Arsie09}
A.~Arsie, K.~Savla, and E.~Frazzoli, ``Efficient routing algorithms for
  multiple vehicles with no explicit communications,'' {\em IEEE Transactions
  on Automatic Control}, vol.~54, no.~10, pp.~2302--2317, 2009.

\bibitem{Huang08}
J.~Huang, Z.~Han, M.~Chiang, and H.~V. Poor, ``Auction-based resource
  allocation for cooperative communications,'' {\em IEEE Journal on Selected
  Areas in Communications,}, vol.~26, no.~7, pp.~1226--1237, 2008.

\bibitem{Jia02}
L.~Jia, R.~Rajaraman, and T.~Suel, ``An efficient distributed algorithm for
  constructing small dominating sets,'' {\em Distributed Computing}, vol.~15,
  no.~4, pp.~193--205, 2002.

\bibitem{Kuhn05}
F.~Kuhn and R.~Wattenhofer, ``Constant-time distributed dominating set
  approximation,'' {\em Distributed Computing}, vol.~17, no.~4, pp.~303--310,
  2005.

\bibitem{Abrams04}
Z.~Abrams, A.~Goel, and S.~Plotkin, ``Set k-cover algorithms for energy
  efficient monitoring in wireless sensor networks,'' in {\em International
  Symposium on Information Processing in Sensor Networks}, pp.~424--432, 2004.

\bibitem{Blume93}
L.~E. Blume, ``The statistical mechanics of strategic interaction,'' {\em Games
  and Economic Behavior}, vol.~5, no.~3, pp.~387--424, 1993.

\bibitem{Marden12}
J.~R. Marden and J.~S. Shamma, ``Revisiting log-linear learning: Asynchrony,
  completeness and payoff-based implementation,'' {\em Games and Economic
  Behavior}, vol.~75, no.~2, pp.~788--808, 2012.

\bibitem{Young93}
H.~P. Young, ``The evolution of conventions,'' {\em Econometrica: Journal of
  the Econometric Society}, vol.~61, no.~1, pp.~57--84, 1993.

\bibitem{Boyd06}
S.~Boyd, A.~Ghosh, B.~Prabhakar, and D.~Shah, ``Randomized gossip algorithms,''
  {\em IEEE Transactions on Information Theory}, vol.~52, no.~6,
  pp.~2508--2530, 2006.

\bibitem{Tarjan72}
R.~Tarjan, ``Depth-first search and linear graph algorithms,'' {\em SIAM
  Journal on Computing}, vol.~1, no.~2, pp.~146--160, 1972.

\end{thebibliography}

\end{document}